\documentclass[submission,copyright,creativecommons]{eptcs}
\providecommand{\event}{Non-Classical Logics 2024 (NCL'24)}

\usepackage{proof}
\usepackage{amssymb}
\usepackage{amsmath,amsthm}
\usepackage{color}

\newtheorem{theorem}{Theorem} 
\newtheorem{proposition}{Proposition}
\newtheorem{lemma}{Lemma}
\newtheorem{corollary}{Corollary}
\newtheorem{definition}{Definition}
\newtheorem{remark}[definition]{Remark}

\newcommand{\Not}{{\sim}}
\newcommand{\car}{\mathop{\circlearrowright}}

\usepackage{iftex}

\ifpdf
  \usepackage{underscore}         % Only needed if you use pdflatex.
  \usepackage[T1]{fontenc}        % Recommended with pdflatex
\else
  \usepackage{breakurl}           % Not needed if you use pdflatex only.
\fi

\title{A note on Grigoriev and Zaitsev's system {\bf CNL$^2_4$}\thanks{The research by Hitoshi Omori was initially supported by a Sofja Kovalevskaja Award of the Alexander von Humboldt-Foundation, funded by the German Ministry for Education and Research. The research by Jonas R. B. Arenhart is supported by CNPq (Brazilian National Research Counsil). We would like to thank the referees for their careful reading, helpful suggestions, and supportive remarks. %Their comments helped to substantially improve the quality of the paper, by spotting some typos and misconceptions that were overlooked by ourselves. Our eternal gratitude. Make it shine!
}}
\author{Hitoshi Omori
\institute{Graduate School of Information Sciences\\
Tohoku University\\
Sendai, Japan}
\email{hitoshiomori@gmail.com}
\and
Jonas R. B. Arenhart
\institute{Department of Philosophy\\
Federal University of Santa Catarina\\ 
Florianópolis, Brazil}
\email{jonas.becker2@gmail.com}
}
\def\titlerunning{A note on Grigoriev and Zaitsev's system {\bf CNL$^2_4$}}
\def\authorrunning{H. Omori \& J.R.B.Arenhart}

\begin{document}
\maketitle

\begin{abstract}
The present article examines a system of four-valued logic recently introduced by Oleg Grigoriev and Dmitry Zaitsev. In particular, besides other interesting results, we will clarify the connection of this system to related systems developed by Paul Ruet and Norihiro Kamide. By doing so, we discuss two philosophical problems that arise from making such connections quite explicit: first, there is an issue with how to make intelligible the meaning of the connectives and the nature of the truth values involved in the many-valued setting employed --- what we have called `the Haackian theme'. We argue that this can be done in a satisfactory way, when seen according to the classicist's light.  Second, and related to the first problem, there is a complication arising from the fact that the proof system advanced may be made sense of by advancing at least four such different and incompatible readings --- a sharpening of the so-called `Carnap problem'. We make explicit how the problems connect with each other precisely and argue that what results is a kind of underdetermination by the deductive apparatus for the system.
\end{abstract}

\section{Introduction}
%Something brief. \hitoshi{We need a general context for this note. What will be the nice context? We need to think about this. }

By its very nature and purpose, a non-classical system of logic is a system that deviates from classical standards on some regards. Most of us believe we can make some sense of what classical connectives mean, and of what classical logical consequence means. Given that for a long time now classical logic has set the standards for the understanding of connectives and logical consequence, whenever some non-classical system of logic is advanced, questions concerning the meaning of the connectives, and what the logical consequence relation is telling us, come to the front. One interesting way to address these questions was suggested some time ago by Susan Haack \cite[chap.11]{haack1978}:\footnote{For recent developments, see \cite{OmoriISMVL22,Omori2023-theoria}.} most of the mysteries of at least some non-classical systems disappears if one can advance a reasonable `classical-like' reading of the connectives. A similar classically-oriented story may be told for logical consequence and for the understanding of the truth values assumed. Having such readings accounts for the classicists' intelligibility of such systems, although, it must be recognize, it deprives them of much of their revolutionary character. 

Given that background, in this paper, we shall be concerned with a system of four-valued logic recently introduced by Oleg Grigoriev and Dmitry Zaitsev in \cite{Grigoriev2022basic}, called {\bf CNL$^2_4$}. Our plan is to apply the strategy suggested by Susan Haack which we have just described, henceforth the `Haackian strategy', in order to make it comfortable for classicists; we shall discuss also whether such application helps us in shedding some light on the system. As part of the implementation, we shall highlight the fact that the strategy may be implemented in different, incompatible ways, generating a scenario where meaning is not properly fixed by the system of logic under scrutiny. Although this is not new, the flexibility on meaning allowed by this system raises interesting questions we shall also discuss.

%\paragraph{Carnap}
%\hitoshi{Our results: one proof system but at least four intuitive readings of one truth table.}

That leads us to the next major philosophical theme we shall concentrate on. The problem we have just touched on regarding meaning flexibility via our attempts to increase the intelligibility here connect to the so-called \textit{Carnap problem}, the fact that a given proof system for a system of logic does not single out \textit{one unique intended interpretation}. As we have indicated, the problem has some very deep ramifications related to the system to be considered here because, as we shall see, given the results to be presented, the same proof system may have at least four different readings. But that is not all yet: those readings are indeed readings of the same ingredients comprising `the' formal semantics, but they are also different enough to suggest that they are instantiations of different approaches to the very understanding of the workings of truth, falsity, the meaning of the connectives and the consequence relation of the underlying system. As a result, they actually seem to count as radically different semantic understandings for the same deductive system. The outcome of this scenario is that two persons using the same system may be having radically different understandings of the references of the logical apparatus, without disagreeing on what follows from what in the system.

%\subsubsection*{Haack}
%\hitoshi{What we can try: how to view the system in a rather classical perspective. One idea is to make use of the functional completeness and specifically the definability of classical negation, and reformulate the semantics in terms of relational semantics (or Dunn semantics), but let the second value be regarded as something like the additional value in Herzberger/Clemens.}

As we have already mentioned, on our way to address the problem of the meaning of the connectives and the intelligibility of the distinct readings proposed for the proof theory, we shall provide for two possible ways to endow {\bf CNL$^2_4$} with a more or less classical reading. The first one is obtained by providing for direct re-readings of the truth values of the original four-valued system. The second one, which will actually instantiate the Haackian strategy, is directly related to a reformulation of the semantics in terms of relational semantics (or Dunn semantics). That will make completely explicit the use of the two classical truth values, and will also illustrate more clearly the different possible readings available in classical terms. As an additional resource for the classicist, given that negation is one of the most controversial connectives, we appeal to functional completeness and to the definability of classical negation inside the system. So, in a sense, the classical logician can gain intelligibility of the working of the system by appeal to a classical behavior that is also available in {\bf CNL$^2_4$}.
%, we suggest that one can reformulate the semantics in terms of relational semantics (or Dunn semantics), but let the second value be regarded as something like a new additional dimension of meaning; this would be very much in the line of the two-dimensional semantics first presented by Hans Herzberger in XXX.}

%\bigskip

%\subsection{Outline of the paper}
The rest of the paper is structured as follows. After a brief preliminaries in \S\ref{sec:preliminaries} recalling the four-valued semantics for {\bf CNL$^2_4$} explored by Grigoriev and Zaitsev, we add some basic results in \S\ref{sec:basic-results}. Building on these results, we turn to the theme from Carnap in \S\ref{sec:Carnap}. This will be followed by \S\ref{sec:Haack} in which we discuss matters in light of the theme from Haack. We shall add some further reflections in \S\ref{sec:reflections}, and the paper will be concluded by \S\ref{sec:concluding-remarks} with some brief final remarks.

\section{Preliminaries}\label{sec:preliminaries}
The language $\mathcal{L}$ consists of a set $\{ \Not , \land , \lor  \}$ of propositional connectives and a countable set $\mathsf{Prop}$ of propositional variables which we denote by $p$, $q$, etc. We denote by $\mathsf{Form}$ the set of formulas defined as usual in $\mathcal{L}$. We denote a formula of $\mathcal{L}$ by $A$, $B$, $C$, etc. and a set of formulas of $\mathcal{L}$ by $\Gamma$, $\Delta$, $\Sigma$, etc. 

Let us now recall the semantics introduced in \cite{Grigoriev2022basic}. For the purpose of this article, we will slightly change the notation to keep the values free of intuitive readings.

\begin{definition}[Grigoriev \& Zaitsev]
A {\bf CNL$^2_4$}-interpretation of $\mathcal{L}$ is a function $I$ from $\mathsf{Prop}$ to $\{ \mathbf{1}, \mathbf{i}, \mathbf{j}, \mathbf{0}\}$. Given a {\bf CNL$^2_4$}-interpretation $I$, this is extended to a valuation $V$ that assigns every formula a truth value by truth functions depicted in the form of truth tables as follows:

\begin{center}
{\small
\begin{tabular}{c|c}
$A$ & $\Not A$ \\
\hline
$\mathbf{1}$ & $\mathbf{i}$ \\
$\mathbf{i}$ & $\mathbf{0}$ \\
$\mathbf{j}$ & $\mathbf{1}$ \\
$\mathbf{0}$ & $\mathbf{j}$ \\
\end{tabular}
\qquad 
\begin{tabular}{c|cccc}
$A \land B$ & $\mathbf{1}$ & $\mathbf{i}$ & $\mathbf{j}$ & $\mathbf{0}$ \\ \hline
$\mathbf{1}$ & $\mathbf{1}$ & $\mathbf{i}$ & $\mathbf{j}$ & $\mathbf{0}$\\
$\mathbf{i}$ & $\mathbf{i}$ & $\mathbf{i}$ & $\mathbf{0}$ & $\mathbf{0}$\\
$\mathbf{j}$ & $\mathbf{j}$ & $\mathbf{0}$ & $\mathbf{j}$ & $\mathbf{0}$\\
$\mathbf{0}$ & $\mathbf{0}$ & $\mathbf{0}$ & $\mathbf{0}$ & $\mathbf{0}$\\
\end{tabular}
\qquad
\begin{tabular}{c|cccc}
$A \lor B$ & $\mathbf{1}$ & $\mathbf{i}$ & $\mathbf{j}$ & $\mathbf{0}$ \\ \hline
$\mathbf{1}$ & $\mathbf{1}$ & $\mathbf{1}$ & $\mathbf{1}$ & $\mathbf{1}$\\
$\mathbf{i}$ & $\mathbf{1}$ & $\mathbf{i}$ & $\mathbf{1}$ & $\mathbf{i}$\\
$\mathbf{j}$ & $\mathbf{1}$ & $\mathbf{1}$ & $\mathbf{j}$ & $\mathbf{j}$\\
$\mathbf{0}$ & $\mathbf{1}$ & $\mathbf{i}$ & $\mathbf{j}$ & $\mathbf{0}$
\end{tabular}
}
\end{center}
\end{definition}

\begin{definition}
For all $\Gamma\cup \{ A \}\subseteq \mathsf{Form}$, $\Gamma\models_{\bf CNL^2_4} A$ iff for all {\bf CNL$^2_4$}-interpretations $I$, $V(A)\in \mathcal{D}$ if $V(B)\in \mathcal{D}$ for all $B\in \Gamma$ where  $\mathcal{D}=\{ \mathbf{1}, \mathbf{i} \}$.
\end{definition}

\begin{remark}
Note that Grigoriev and Zaitsev also consider another four-valued logic called {\bf CNLL$^2_4$} in which four values are \emph{linearly ordered}. We shall not, however, consider the other system since there are already plenty of topics to discuss for {\bf CNL$^2_4$}.
\end{remark}

\section{Basic observations}\label{sec:basic-results}

\subsection{An alternative proof system}
In \cite{Grigoriev2022basic}, a binary proof system is defined by Grigoriev and Zaitsev, but here we will present a natural deduction system.

\begin{definition}
The natural deduction rules $\mathcal{R}_{\bf CNL^2_4}$ for {\bf CNL$^2_4$} are all the following rules:
$$
\infer{A {\wedge}B}{A & B}
\quad
\infer{A}{ A {\wedge}B }
\quad
\infer{B}{ A {\wedge}B }
\quad 
\infer{A {\vee}B}{A}
\quad
\infer{A {\vee}B}{B}
\quad
\infer{C}
{
A {\vee}B
&
\infer*{C}{[A]}
&
\infer*{C}{[B]}
}
\qquad
\infer[(\Not\Not 1)]{B}{A \quad \Not\Not A}
\qquad
\infer[(\Not\Not 2)]{ A{\vee}\Not\Not A}{}
$$
\[
\infer{\Not (A {\land}B)}{\Not A & \Not B}
\qquad
\infer{\Not A}{ \Not (A {\land}B) }
\qquad
\infer{\Not B}{ \Not (A {\land}B) }
\qquad
\infer{\Not (A {\lor}B)}{\Not A}
\qquad
\infer{\Not (A {\lor}B)}{\Not B}
\qquad
\infer{C}
{
\Not (A {\lor}B)
&
\infer*{C}{[\Not A]}
&
\infer*{C}{[\Not B]}
}
\]
Based on these, given any set $\Sigma \cup \{ A \}$ of formulas, $\Sigma \vdash A$ iff for some finite $\Sigma'\subseteq \Sigma$, there is a derivation of $A$ from $\Sigma'$ in the calculus whose rule set is $\mathcal{R}_{\bf CNL^2_4}$.
\end{definition}

Then, the soundness direction is tedious, but standard, so we only state it without a proof.

\begin{theorem}[Soundness]
For all $\Gamma\cup \{ A \}\subseteq \mathsf{Form}$, $\Gamma\vdash A$ only if $\Gamma\models_{\bf CNL^2_4}A$.
\end{theorem}
%\begin{proof}
%Tedious, but standard. 
%\end{proof}

For the completeness direction, we prepare some well known notions and lemmas.
\begin{definition}
Let $\Sigma$ be a set of formulas. Then, 
$\Sigma$ is a \emph{theory} iff $\Sigma \vdash A$ implies $A \in\Sigma$, and $\Sigma$ is \emph{prime} iff $A{\vee}B \in\Sigma$ implies $A\in\Sigma$ or $B\in\Sigma$.
\end{definition}

\begin{lemma}[Lindenbaum]\label{lem:Lindenbaum}
If $\Sigma\not\vdash A$, then there is $\Sigma'\supseteq \Sigma$ such that $\Sigma'\not\vdash A$ and $\Sigma'$ is a prime theory.
\end{lemma}

We now define the canonical valuation in the following manner.
\begin{definition}\label{def:canonical-valuation1}
For any $\Sigma\subseteq \mathsf{Form}$, let $v_\Sigma$ from $\mathsf{Prop}$ to $\{\mathbf{1}, \mathbf{i}, \mathbf{j}, \mathbf{0}\}$ be defined as follows:

{\small
\begin{center} 
$v_\Sigma (p):=
\begin{cases}
\mathbf{1} \text{ iff } \Sigma\vdash p \text{ and } \Sigma\vdash \Not p; \\
\mathbf{i}  \text{ iff } \Sigma\vdash p \text{ and } \Sigma\not\vdash \Not p; \\
\mathbf{j}  \text{ iff } \Sigma\not\vdash p \text{ and } \Sigma\vdash \Not p;\\
\mathbf{0} \text{ iff } \Sigma\not\vdash p \text{ and } \Sigma\not\vdash \Not p.
\end{cases}
$
\end{center}
}
\end{definition}

\begin{remark}
Note that the above definition is different from the more familiar definition when the four values are understood as in {\bf FDE}.
\end{remark}

The following lemma is the key for the completeness result.
\begin{lemma}\label{lem:canonical-valuation1} 
If $\Sigma$ is a prime theory, then the following hold for all $B\in\mathsf{Form}$.

{\small
\begin{center}
$v_\Sigma (B)=
\begin{cases}
\mathbf{1} \text{ iff } \Sigma\vdash B \text{ and } \Sigma\vdash \Not B;\\
\mathbf{i}  \text{ iff } \Sigma\vdash B \text{ and } \Sigma\not\vdash \Not B; \\
\mathbf{j}  \text{ iff } \Sigma\not\vdash B \text{ and } \Sigma\vdash \Not B; \\ 
\mathbf{0} \text{ iff } \Sigma\not\vdash B \text{ and } \Sigma\not\vdash \Not B.
\end{cases}
$
\end{center}
}
\end{lemma}
\begin{proof}
Note first that the well-definedness of $v_\Sigma$ is obvious. Then the desired result is proved by induction on the construction of $B$. The base case, for atomic formulas, is obvious by the definition. For the induction step, the cases are split based on the connectives. The details are spelled out in the Appendix.
\end{proof}

We are now ready to prove the completeness result.
\begin{theorem}[Completeness]\label{thm:completeness-for-unilateral}
For all $\Gamma\cup \{ A \}\subseteq \mathsf{Form}$, $\Gamma\models_{\bf CNL^2_4} A$ only if $\Gamma\vdash A$.
\end{theorem}
\begin{proof}
Assume $\Gamma\not\vdash A$. Then, by Lemma~\ref{lem:Lindenbaum}, there is a $\Sigma\supseteq \Gamma$ such that $\Sigma$ is a prime theory and $A\not\in\Sigma$, and by Lemma~\ref{lem:canonical-valuation1}, a four-valued valuation $v_{\Sigma}$ can be defined with $I_{\Sigma}(B)\in\mathcal{D}$ for every $B\in\Gamma$ and $I_{\Sigma}(A)\not\in\mathcal{D}$. Thus it follows that $\Gamma \not\models_{\bf CNL^2_4} A$, as desired.
\end{proof}

\subsection{Functional completeness}
We now turn to show that the matrix that characterizes the system is functionally complete. To this end, we will first introduce some related notions.

\begin{definition}[Functional completeness]
An algebra $\mathfrak{A}=\langle A, f_1, \dots , f_n \rangle$, is said to be \emph{functionally complete} provided that every finitary function ${f\!:}$ $A^m \to A$ is definable by compositions of the functions $f_1, \dots, f_n$ alone. A matrix $\langle \mathfrak{A}, {\mathcal D} \rangle$ is \emph{functionally complete} if $\mathfrak{A}$ is functionally complete.
\end{definition}

\begin{definition}[Definitional completeness]
A logic {\bf L} is \emph{definitionally complete} if there exists a functionally complete matrix that is strongly sound and complete for $L$.
\end{definition}

For the characterization of the functional completeness, the following theorem of Jerzy S{\l}upecki is elegant and useful. In order to state the result, we need the following definition.

\begin{definition}\label{def:essen.binary}
Let $\mathfrak{A}$ $=$ $\langle A, f_1, \dots , f_n \rangle$ be an algebra, and $f$ be a binary operation defined in $\mathfrak{A}$. Then, $f$ is \emph{unary reducible} iff for some unary operation $g$ definable in $\mathfrak{A}$, $f(x, y)=g(x)$ for all $x, y\in A$ or $f(x, y)=g(y)$ for all $x, y\in A$. And $f$ is \emph{essentially binary} if $f$ is not unary reducible.
\end{definition}

\begin{theorem}[S{\l}upecki, \cite{Slupecki1972}]\label{thm:Slupecki}
$\mathfrak{A}$ $= \langle \langle {\mathcal V}, f_1, \dots , f_n \rangle, {\mathcal D}\rangle$ $( |\mathcal{V}|\geq 3)$ is functionally complete iff in\linebreak $\langle {\mathcal V}, f_1, \dots , f_n \rangle$
{\rm (1)} all unary functions on $\mathcal{V}$ are definable, and {\rm (2)} at least one surjective and essentially binary function on $\mathcal{V}$ is definable.
\end{theorem}

This elegant characterization by S{\l}upecki can be simplified even further in case of expansions of the algebra related to {\bf FDE} (cf. \cite[Theorem 4.8]{HH2018contra}).

\begin{theorem}\label{thm:f.c.}
Given any expansion $\mathcal{F}$ of the algebra $\langle \{\mathbf{1}, \mathbf{i}, \mathbf{j}, \mathbf{0}\}, \wedge, {\vee}\rangle$, the following are equivalent: 
\begin{itemize}
\setlength{\parskip}{0cm}
\setlength{\itemsep}{0cm}
\item[{\rm (1)}] $\mathcal{F}$ is functionally complete; 
\item[{\rm (2)}] all of the $\delta_a$s as well as
$\mathsf{C}_a$s {\rm($a\in \{ {\mathbf 1}, {\mathbf i}, {\mathbf j}, {\mathbf 0} \}$)} are definable, where $\delta_a(b){=}\mathbf{1}$, if $a{=}b$, otherwise $\delta_a(b){=}\mathbf{0}$; and $\mathsf{C}_a(b){=}a$, for all $a,b\in\mathcal{V}$.
\end{itemize}
\end{theorem}

Building on this result, we obtain the following.
\begin{theorem}\label{thm:def.comp}
{\bf CNL$^2_4$} is definitionally complete.
\end{theorem}
\vspace{-3mm}
\begin{proof}
In view of the above theorem, it suffices to prove that all of the $\delta_a$s as well as $\mathsf{C}_a$s {\rm($a\in \{ {\mathbf 1}, {\mathbf i}, {\mathbf j}, {\mathbf 0} \}$)} are definable in $\langle \{\mathbf{1}, {\mathbf i}, {\mathbf j}, \mathbf{0}\}, \Not, \wedge, {\vee}\rangle$, and this can be done as follows: 

\smallskip

\noindent 
\begin{tabular}{llll}
$\delta_\mathbf{1}(x){:=}\neg (\Not \Not x {\vee} \Not \Not \Not x)$,&\!\!\!\!\!\!
$\delta_\mathbf{j}(x){:=}\neg (\Not \Not \Not x {\vee} \neg \neg  x)$,&\!\!\!\!\!\!
$\mathsf{C}_\mathbf{1}(x){:=}x {\lor} \Not \Not x$,&\!\!\!\!\!\!
$\mathsf{C}_\mathbf{j}(x){:=}\Not (x {\land} \Not \Not x)$, \\
$\delta_\mathbf{i}(x){:=}\neg (\Not \Not x {\vee} \neg \Not \Not \Not x)$, & 
$\delta_\mathbf{0}(x){:=}\neg \neg (\Not \Not x \wedge \Not \Not \Not x)$, &
$\mathsf{C}_\mathbf{i}(x){:=}\Not (x {\lor} \Not \Not x)$, &
$\mathsf{C}_\mathbf{0}(x){:=}x {\land} \Not \Not x$,
\end{tabular}
\medskip
\noindent where $\neg x{:=} \Not \Not \Not (\Not\Not ((x{\wedge}\Not \Not \Not x){\wedge}\Not \Not \Not (x{\wedge}\Not \Not \Not x)){\wedge}((x{\wedge}\Not \Not \Not x){\vee}\Not \Not \Not (x{\wedge}\Not \Not \Not x)))$. 
%This completes the proof. 
\end{proof}

Finally, we add a brief remark on the Post completeness.

\begin{definition}
The logic {\bf L} is \emph{Post complete} iff for every formula $A$ such that $\not\vdash A$, the extension of {\bf L} by $A$ becomes trivial, i.e., $\vdash_{{\bf L}\cup \{ A \} } B$ for any $B$.
\end{definition}

\begin{theorem}[Tokarz, \cite{Tokarz1973}]\label{thm:Tokarz}
Definitionally complete logics are Post complete.
\end{theorem}

In view of Theorems~\ref{thm:def.comp} and~\ref{thm:Tokarz}, we obtain the following result.

\begin{corollary}
{\bf CNL$^2_4$} is Post complete.
\end{corollary}

\subsection{A few more results}
Before moving further, we list some  valid/derivable inferences, as well as invalid/non-derivable ones.
\begin{proposition}\label{prop:valid-inferences}
The following hold in {\bf CNL$^2_4$}.

\smallskip
$B\models_{\bf CNL^2_4} (A{\lor}\Not\Not A)$, \quad
$B\models_{\bf CNL^2_4} \Not (A{\lor}\Not\Not A)$, \quad
$A{\land}\Not\Not A\models_{\bf CNL^2_4} B$, \quad
$\Not (A{\land}\Not\Not A)\models_{\bf CNL^2_4} B$.
\end{proposition}

\vspace{-3mm}
\begin{proof}
It suffices to observe that $V(A{\lor}\Not\Not A){=}\mathbf{1}$, and $V(\Not(A{\lor}\Not\Not A)){=}\mathbf{i}$ for the first two items, and that $V(A{\land}\Not\Not A){=}\mathbf{0}$, and $V(\Not(A{\land}\Not\Not A)){=}\mathbf{j}$ for the latter two items.
\end{proof}

\begin{proposition}
The following also hold in {\bf CNL$^2_4$}.

\smallskip
$q\not\models_{\bf CNL^2_4} p{\lor}\Not p$, \quad
$p{\land}\Not p\not\models_{\bf CNL^2_4} q$, \quad
$\Not \Not p\not\models_{\bf CNL^2_4} p$, \quad
$p\not\models_{\bf CNL^2_4} \Not \Not p$.
\end{proposition}

\vspace{-3mm}
\begin{proof}
Interpretations such that $I_1(p){=}\mathbf{0}, I_2(q){=}\mathbf{1}$ for the first item, $I_2(p){=}\mathbf{1}, I_2(q){=}\mathbf{0}$ for the second item, $I_3(p){=}\mathbf{0}$ for the third item, and $I_4(p){=}\mathbf{1}$ for the last item will establish the desired results.
\end{proof}

\begin{remark}
One might be already tempted to discuss features of {\bf CNL$^2_4$} based on the above observations. In particular, one may be tempted to refer to $\Not$ as negation. This, however, is a rather delicate matter, and we will return to this point in \S\ref{sec:reflections} after some discussions on the interpretations of the four values.
\end{remark}

\section{Carnapian theme: Four interpretations of one truth table}\label{sec:Carnap}

We now present the options for the readings of the truth values, according to the two strategies we mentioned before, viz., the re-readings for the original four-valued semantics, and the relational semantics. 

\subsection{Option 1}
The first two options will be to interpret $\mathbf{1}$ and $\mathbf{0}$ as $\mathbf{t}$ and $\mathbf{f}$ of {\bf FDE}, respectively, and make a choice between options in interpreting the intermediate values. Let us start by following the choice made by Grigoriev and Zaitsev, that is, the values $\mathbf{i}$ and $\mathbf{j}$ are interpreted as $\mathbf{b}$ and $\mathbf{n}$, respectively. 

\begin{center}
{\small
\begin{tabular}{c|c}
$A$ & $\Not A$ \\
\hline
$\mathbf{t}$ & $\mathbf{b}$ \\
$\mathbf{b}$ & $\mathbf{f}$ \\
$\mathbf{n}$ & $\mathbf{t}$ \\
$\mathbf{f}$ & $\mathbf{n}$ \\
\end{tabular}
\qquad 
\begin{tabular}{c|cccc}
$A \land B$ & $\mathbf{t}$ & $\mathbf{b}$ & $\mathbf{n}$ & $\mathbf{f}$ \\ \hline
$\mathbf{t}$ & $\mathbf{t}$ & $\mathbf{b}$ & $\mathbf{n}$ & $\mathbf{f}$\\
$\mathbf{b}$ & $\mathbf{b}$ & $\mathbf{b}$ & $\mathbf{f}$ & $\mathbf{f}$\\
$\mathbf{n}$ & $\mathbf{n}$ & $\mathbf{f}$ & $\mathbf{n}$ & $\mathbf{f}$\\
$\mathbf{f}$ & $\mathbf{f}$ & $\mathbf{f}$ & $\mathbf{f}$ & $\mathbf{f}$\\
\end{tabular}
\qquad
\begin{tabular}{c|cccc}
$A \lor B$ & $\mathbf{t}$ & $\mathbf{b}$ & $\mathbf{n}$ & $\mathbf{f}$ \\ \hline
$\mathbf{t}$ & $\mathbf{t}$ & $\mathbf{t}$ & $\mathbf{t}$ & $\mathbf{t}$\\
$\mathbf{b}$ & $\mathbf{t}$ & $\mathbf{b}$ & $\mathbf{t}$ & $\mathbf{b}$\\
$\mathbf{n}$ & $\mathbf{t}$ & $\mathbf{t}$ & $\mathbf{n}$ & $\mathbf{n}$\\
$\mathbf{f}$ & $\mathbf{t}$ & $\mathbf{b}$ & $\mathbf{n}$ & $\mathbf{f}$
\end{tabular}
}
\end{center}
Then, this makes it very clear that the resulting truth tables are those introduced by Paul Ruet in \cite{Ruet1996}. Moreover, the resulting logic is obtained by considering the truth preservation by building on the above truth tables, and therefore, it is the same logic introduced by Ruet.\footnote{To be more precise, Ruet added the unary operator $\Not$ (or $\car$ in Ruet's notation) on top of Belnap-Dunn logic. Therefore, for the purpose of establishing the definitional equivalence of the system of Ruet and the system of Grigoriev and Zaitsev, we need to check that de Morgan negation is definable in {\bf CNL$^2_4$}. However, this is an immediate corollary of the functional completeness result. Therefore, the desired result is established.}\footnote{Note that Ruet's system is also discussed in \cite{BGZ2022} under the name {\bf dCP} by Grigoriev and Zaitsev together with Alex Belikov. Even though Grigoriev and Zaitsev do not state it explicitly in \cite{Grigoriev2022basic}, {\bf CNL$^2_4$} is definitionally equivalent to {\bf dCP}.}

In order to observe the differences of interpretations, let us apply the mechanical procedure described in \cite{OS-BD-FC}, and offer an alternative presentation of the interpretation in terms of truth and falsity conditions, assuming that we rewrite the four values $\mathbf{t}, \mathbf{b}, \mathbf{n}$ and $\mathbf{f}$ as $\{ 1 \}, \{ 1, 0 \}, \emptyset$ and $\{ 0 \}$, respectively. For the present case, we obtain the following truth and falsity conditions. 

\begin{tabular}{ll}
$1\in V(\Not A)$ iff $0\not\in V(A)$; & $0\in V(\Not A)$ iff $1\in V(A)$;\\
$1\in V(A{\land} B)$ iff $1\in V(A)$ and $1\in V(B)$; & $0\in V(A{\land} B)$ iff $0\in V(A)$ or $0\in V(B)$; \\
$1\in V(A{\lor} B)$ iff $1\in V(A)$ or $1\in V(B)$; & $0\in V(A{\lor} B)$ iff $0\in V(A)$ and $0\in V(B)$.
\end{tabular}

\noindent 
Therefore, it becomes very clear that the truth and falsity conditions are almost the same with {\bf FDE}. Indeed, $\land$ and $\lor$ are interpreted as in {\bf FDE}, and for $\Not$, the truth condition is the only condition that is deviating from {\bf FDE}.\footnote{Note that $\Not\Not$ behaves as the Boolean complement. For discussions on such a kind of connective, see \cite{Humberstone1995,Kamide2016, Kamide40FDE,HH2018contra,PaoliDemiNegation,Omori2022varieties}.}

\subsection{Option 2}
Let us now turn to the other option. That is, the values $\mathbf{i}$ and $\mathbf{j}$ are interpreted as $\mathbf{n}$ and $\mathbf{b}$, respectively. Then, as a result of rewriting the values, we obtain the following truth table.

\begin{center}
{\small
\begin{tabular}{c|c}
$A$ & $\Not A$ \\
\hline
$\mathbf{t}$ & $\mathbf{n}$ \\
$\mathbf{n}$ & $\mathbf{f}$ \\
$\mathbf{b}$ & $\mathbf{t}$ \\
$\mathbf{f}$ & $\mathbf{b}$ \\
\end{tabular}
\qquad 
\begin{tabular}{c|cccc}
$A \land B$ & $\mathbf{t}$ & $\mathbf{n}$ & $\mathbf{b}$ & $\mathbf{f}$ \\ \hline
$\mathbf{t}$ & $\mathbf{t}$ & $\mathbf{n}$ & $\mathbf{b}$ & $\mathbf{f}$\\
$\mathbf{n}$ & $\mathbf{n}$ & $\mathbf{n}$ & $\mathbf{f}$ & $\mathbf{f}$\\
$\mathbf{b}$ & $\mathbf{b}$ & $\mathbf{f}$ & $\mathbf{b}$ & $\mathbf{f}$\\
$\mathbf{f}$ & $\mathbf{f}$ & $\mathbf{f}$ & $\mathbf{f}$ & $\mathbf{f}$\\
\end{tabular}
\qquad
\begin{tabular}{c|cccc}
$A \lor B$ & $\mathbf{t}$ & $\mathbf{n}$ & $\mathbf{b}$ & $\mathbf{f}$ \\ \hline
$\mathbf{t}$ & $\mathbf{t}$ & $\mathbf{t}$ & $\mathbf{t}$ & $\mathbf{t}$\\
$\mathbf{n}$ & $\mathbf{t}$ & $\mathbf{n}$ & $\mathbf{t}$ & $\mathbf{n}$\\
$\mathbf{b}$ & $\mathbf{t}$ & $\mathbf{t}$ & $\mathbf{b}$ & $\mathbf{b}$\\
$\mathbf{f}$ & $\mathbf{t}$ & $\mathbf{n}$ & $\mathbf{b}$ & $\mathbf{f}$
\end{tabular}
}
\end{center}
If we rewrite it slightly, for the purpose of making the comparison easier, we obtain the following tables. 

\begin{center}
{\small
\begin{tabular}{c|c}
$A$ & $\Not A$ \\
\hline
$\mathbf{t}$ & $\mathbf{n}$ \\
$\mathbf{b}$ & $\mathbf{t}$ \\
$\mathbf{n}$ & $\mathbf{f}$ \\
$\mathbf{f}$ & $\mathbf{b}$ \\
\end{tabular}
\qquad 
\begin{tabular}{c|cccc}
$A \land B$ & $\mathbf{t}$ & $\mathbf{b}$ & $\mathbf{n}$ & $\mathbf{f}$ \\ \hline
$\mathbf{t}$ & $\mathbf{t}$ & $\mathbf{b}$ & $\mathbf{n}$ & $\mathbf{f}$\\
$\mathbf{b}$ & $\mathbf{b}$ & $\mathbf{b}$ & $\mathbf{f}$ & $\mathbf{f}$\\
$\mathbf{n}$ & $\mathbf{n}$ & $\mathbf{f}$ & $\mathbf{n}$ & $\mathbf{f}$\\
$\mathbf{f}$ & $\mathbf{f}$ & $\mathbf{f}$ & $\mathbf{f}$ & $\mathbf{f}$\\
\end{tabular}
\qquad
\begin{tabular}{c|cccc}
$A \lor B$ & $\mathbf{t}$ & $\mathbf{b}$ & $\mathbf{n}$ & $\mathbf{f}$ \\ \hline
$\mathbf{t}$ & $\mathbf{t}$ & $\mathbf{t}$ & $\mathbf{t}$ & $\mathbf{t}$\\
$\mathbf{b}$ & $\mathbf{t}$ & $\mathbf{b}$ & $\mathbf{t}$ & $\mathbf{b}$\\
$\mathbf{n}$ & $\mathbf{t}$ & $\mathbf{t}$ & $\mathbf{n}$ & $\mathbf{n}$\\
$\mathbf{f}$ & $\mathbf{t}$ & $\mathbf{b}$ & $\mathbf{n}$ & $\mathbf{f}$
\end{tabular}
}
\end{center}
Then, this makes it clear that the resulting truth tables are those introduced by Norihiro Kamide in \cite{Kamide2016,Kamide40FDE}, and explored in \cite{HH2018contra,Omori2022varieties,PaoliDemiNegation}.\footnote{To be precise, there are some differences in the language. Indeed, in \cite{Kamide2016,Kamide40FDE,HH2018contra}, a classical conditional is added, while informational join and meet are added in \cite{PaoliDemiNegation}. The language in \cite{Omori2022varieties} is the same as here.} Moreover, we obtain the following truth and falsity conditions. 

\begin{tabular}{ll}
$1\in V(\Not A)$ iff $0\in V(A)$; & $0\in V(\Not A)$ iff $1\not\in V(A)$; \\
$1\in V(A{\land} B)$ iff $1\in V(A)$ and $1\in V(B)$; & $0\in V(A{\land} B)$ iff $0\in V(A)$ or $0\in V(B)$; \\
$1\in V(A{\lor} B)$ iff $1\in V(A)$ or $1\in V(B)$; & $0\in V(A{\lor} B)$ iff $0\in V(A)$ and $0\in V(B)$.
\end{tabular}

\noindent Compared to {\bf FDE}, the only difference lies in the falsity conditions for $\Not$.

Note, however, that the resulting logic is \emph{not} the same since the designated values are $\mathbf{t}$ and $\mathbf{n}$, not $\mathbf{t}$ and $\mathbf{b}$. In other words, we are considering the consequence relation in terms of non-falsity preservation, rather than truth preservation.\footnote{The case of consequence relation defined in terms of truth preservation within the same language is explored in \cite{Omori2022varieties}.}

\subsection{Option 3}
Seen in the light of the natural deduction system, we may also think of regarding the binary connectives as \emph{information} connectives, rather than \emph{truth} connectives. This will correspond to interpret $\mathbf{1}$ and $\mathbf{0}$ as $\mathbf{b}$ and $\mathbf{n}$ of {\bf FDE}, respectively. Then, there are again two options in interpreting the intermediate values. Let us begin with the case in which we interpret $\mathbf{i}$ and $\mathbf{j}$ as $\mathbf{t}$ and $\mathbf{f}$ of {\bf FDE}, respectively. Then, as a result of rewriting the values, we obtain the following truth table.

\begin{center}
{\small
\begin{tabular}{c|c}
$A$ & $\Not A$ \\
\hline
$\mathbf{b}$ & $\mathbf{t}$ \\
$\mathbf{t}$ & $\mathbf{n}$ \\
$\mathbf{f}$ & $\mathbf{b}$ \\
$\mathbf{n}$ & $\mathbf{f}$ \\
\end{tabular}
\qquad 
\begin{tabular}{c|cccc}
$A \land B$ & $\mathbf{b}$ & $\mathbf{t}$ & $\mathbf{f}$ & $\mathbf{n}$ \\ \hline
$\mathbf{b}$ & $\mathbf{b}$ & $\mathbf{t}$ & $\mathbf{f}$ & $\mathbf{n}$\\
$\mathbf{t}$ & $\mathbf{t}$ & $\mathbf{t}$ & $\mathbf{n}$ & $\mathbf{n}$\\
$\mathbf{f}$ & $\mathbf{f}$ & $\mathbf{n}$ & $\mathbf{f}$ & $\mathbf{n}$\\
$\mathbf{n}$ & $\mathbf{n}$ & $\mathbf{n}$ & $\mathbf{n}$ & $\mathbf{n}$\\
\end{tabular}
\qquad
\begin{tabular}{c|cccc}
$A \lor B$ & $\mathbf{b}$ & $\mathbf{t}$ & $\mathbf{f}$ & $\mathbf{n}$ \\ \hline
$\mathbf{b}$ & $\mathbf{b}$ & $\mathbf{b}$ & $\mathbf{b}$ & $\mathbf{b}$\\
$\mathbf{t}$ & $\mathbf{b}$ & $\mathbf{t}$ & $\mathbf{b}$ & $\mathbf{t}$\\
$\mathbf{f}$ & $\mathbf{b}$ & $\mathbf{b}$ & $\mathbf{f}$ & $\mathbf{f}$\\
$\mathbf{n}$ & $\mathbf{b}$ & $\mathbf{t}$ & $\mathbf{f}$ & $\mathbf{n}$
\end{tabular}
}
\end{center}
If we again rewrite it slightly, then we obtain the following truth tables. 

\begin{center}
{\small
\begin{tabular}{c|c}
$A$ & $\Not A$ \\
\hline
$\mathbf{t}$ & $\mathbf{n}$ \\
$\mathbf{b}$ & $\mathbf{t}$ \\
$\mathbf{n}$ & $\mathbf{f}$ \\
$\mathbf{f}$ & $\mathbf{b}$ \\
\end{tabular}
\qquad 
\begin{tabular}{c|cccc}
$A \land B$ & $\mathbf{t}$ & $\mathbf{b}$ & $\mathbf{n}$ & $\mathbf{f}$ \\
\hline
$\mathbf{t}$ & $\mathbf{t}$ & $\mathbf{t}$ & $\mathbf{n}$ & $\mathbf{n}$\\
$\mathbf{b}$ & $\mathbf{t}$ & $\mathbf{b}$ & $\mathbf{n}$ & $\mathbf{f}$\\
$\mathbf{n}$ & $\mathbf{n}$ & $\mathbf{n}$ & $\mathbf{n}$ & $\mathbf{n}$\\
$\mathbf{f}$ & $\mathbf{n}$ & $\mathbf{f}$ & $\mathbf{n}$ & $\mathbf{f}$\\
\end{tabular}
\qquad
\begin{tabular}{c|cccc}
$A \lor B$ & $\mathbf{t}$ & $\mathbf{b}$ & $\mathbf{n}$ & $\mathbf{f}$ \\
\hline
$\mathbf{t}$ & $\mathbf{t}$ & $\mathbf{b}$ & $\mathbf{t}$ & $\mathbf{b}$\\
$\mathbf{b}$ & $\mathbf{b}$ & $\mathbf{b}$ & $\mathbf{b}$ & $\mathbf{b}$\\
$\mathbf{n}$ & $\mathbf{t}$ & $\mathbf{b}$ & $\mathbf{n}$ & $\mathbf{f}$\\
$\mathbf{f}$ & $\mathbf{b}$ & $\mathbf{b}$ & $\mathbf{f}$ & $\mathbf{f}$\\
\end{tabular}
}
\end{center}
Then, this makes it clear that the resulting truth tables are obtained by putting Kamide's unary operator together with information meet and join connectives. Moreover, we obtain the following truth and falsity conditions. 

\begin{tabular}{ll}
$1\in V(\Not A)$ iff $0\in V(A)$; & $0\in V(\Not A)$ iff $1\not\in V(A)$; \\
$1\in V(A{\land} B)$ iff $1\in V(A)$ and $1\in V(B)$; & $0\in V(A{\land} B)$ iff $0\in V(A)$ and $0\in V(B)$; \\
$1\in V(A{\lor} B)$ iff $1\in V(A)$ or $1\in V(B)$; & $0\in V(A{\lor} B)$ iff $0\in V(A)$ or $0\in V(B)$.
\end{tabular}

\noindent Now, compared to {\bf FDE}, the truth conditions are exactly the same for all the connectives. However, the falsity conditions are different, and in particular, for $\land$ and $\lor$, those are taken as in the information meet and join, respectively.

Finally, since the designated values are $\mathbf{t}$ and $\mathbf{b}$, the resulting logic is obtained by considering the truth preservation building on the above truth tables.

\subsection{Option 4}
Let us now turn to the other option. That is, we interpret $\mathbf{i}$ and $\mathbf{j}$ as $\mathbf{f}$ and $\mathbf{t}$ of {\bf FDE}, respectively. Then, as a result of rewriting the values, we obtain the following truth table.

\begin{center}
{\small
\begin{tabular}{c|c}
$A$ & $\Not A$ \\
\hline
$\mathbf{b}$ & $\mathbf{f}$ \\
$\mathbf{f}$ & $\mathbf{n}$ \\
$\mathbf{t}$ & $\mathbf{b}$ \\
$\mathbf{n}$ & $\mathbf{t}$ \\
\end{tabular}
\qquad 
\begin{tabular}{c|cccc}
$A \land B$ & $\mathbf{b}$ & $\mathbf{f}$ & $\mathbf{t}$ & $\mathbf{n}$ \\ \hline
$\mathbf{b}$ & $\mathbf{b}$ & $\mathbf{f}$ & $\mathbf{t}$ & $\mathbf{n}$\\
$\mathbf{f}$ & $\mathbf{f}$ & $\mathbf{t}$ & $\mathbf{n}$ & $\mathbf{n}$\\
$\mathbf{t}$ & $\mathbf{t}$ & $\mathbf{n}$ & $\mathbf{t}$ & $\mathbf{n}$\\
$\mathbf{n}$ & $\mathbf{n}$ & $\mathbf{n}$ & $\mathbf{n}$ & $\mathbf{n}$\\
\end{tabular}
\qquad
\begin{tabular}{c|cccc}
$A \lor B$ & $\mathbf{b}$ & $\mathbf{f}$ & $\mathbf{t}$ & $\mathbf{n}$ \\ \hline
$\mathbf{b}$ & $\mathbf{b}$ & $\mathbf{b}$ & $\mathbf{b}$ & $\mathbf{b}$\\
$\mathbf{f}$ & $\mathbf{b}$ & $\mathbf{f}$ & $\mathbf{b}$ & $\mathbf{f}$\\
$\mathbf{t}$ & $\mathbf{b}$ & $\mathbf{b}$ & $\mathbf{t}$ & $\mathbf{t}$\\
$\mathbf{n}$ & $\mathbf{b}$ & $\mathbf{f}$ & $\mathbf{t}$ & $\mathbf{n}$
\end{tabular}
}
\end{center}
If we again rewrite it slightly, then we obtain the following truth tables. 

\begin{center}
{\small
\begin{tabular}{c|c}
$A$ & $\Not A$ \\
\hline
$\mathbf{t}$ & $\mathbf{b}$ \\
$\mathbf{b}$ & $\mathbf{f}$ \\
$\mathbf{n}$ & $\mathbf{t}$ \\
$\mathbf{f}$ & $\mathbf{n}$ \\
\end{tabular}
\qquad 
\begin{tabular}{c|cccc}
$A \land B$ & $\mathbf{t}$ & $\mathbf{b}$ & $\mathbf{n}$ & $\mathbf{f}$ \\
\hline
$\mathbf{t}$ & $\mathbf{t}$ & $\mathbf{t}$ & $\mathbf{n}$ & $\mathbf{n}$\\
$\mathbf{b}$ & $\mathbf{t}$ & $\mathbf{b}$ & $\mathbf{n}$ & $\mathbf{f}$\\
$\mathbf{n}$ & $\mathbf{n}$ & $\mathbf{n}$ & $\mathbf{n}$ & $\mathbf{n}$\\
$\mathbf{f}$ & $\mathbf{n}$ & $\mathbf{f}$ & $\mathbf{n}$ & $\mathbf{f}$\\
\end{tabular}
\qquad
\begin{tabular}{c|cccc}
$A \lor B$ & $\mathbf{t}$ & $\mathbf{b}$ & $\mathbf{n}$ & $\mathbf{f}$ \\
\hline
$\mathbf{t}$ & $\mathbf{t}$ & $\mathbf{b}$ & $\mathbf{t}$ & $\mathbf{b}$\\
$\mathbf{b}$ & $\mathbf{b}$ & $\mathbf{b}$ & $\mathbf{b}$ & $\mathbf{b}$\\
$\mathbf{n}$ & $\mathbf{t}$ & $\mathbf{b}$ & $\mathbf{n}$ & $\mathbf{f}$\\
$\mathbf{f}$ & $\mathbf{b}$ & $\mathbf{b}$ & $\mathbf{f}$ & $\mathbf{f}$\\
\end{tabular}
}
\end{center}
Then, this makes it clear that the resulting truth tables are obtained by putting Ruet's unary operator together with information meet and join connectives. Moreover, we obtain the following truth and falsity conditions. 

\begin{tabular}{ll}
$1\in V(\Not A)$ iff $0\not\in V(A)$; & $0\in V(\Not A)$ iff $1\in V(A)$; \\
$1\in V(A{\land} B)$ iff $1\in V(A)$ and $1\in V(B)$; & $0\in V(A{\land} B)$ iff $0\in V(A)$ and $0\in V(B)$; \\
$1\in V(A{\lor} B)$ iff $1\in V(A)$ or $1\in V(B)$; & $0\in V(A{\lor} B)$ iff $0\in V(A)$ or $0\in V(B)$.
\end{tabular}

\noindent Now, compared to {\bf FDE}, the falsity condition is the same for $\Not$, and the truth conditions are exactly the same for $\land$ and $\lor$. However, the other conditions are different, and in particular, for $\land$ and $\lor$, the falsity condition is taken as in the information meet and join, respectively.

Moreover, since the designated values are $\mathbf{b}$ and $\mathbf{f}$, the resulting logic is obtained by considering the falsity preservation building on the above truth tables.

%\hitoshi{Maybe we can add a table with the summary of the four options. Basically, given the truth table, and the natural deduction system, there are at least four different ways to interpret the values/rules for the system. It seems like too many options!!! Perhaps we can relate the result to the categoricity problem. The observation may be strengthened by adding a more general semantics in terms of bilattice-theoretic semantics so that we can say that the problem is much more substantial than in classical logic? This is because even if we pin down the number of the semantic values, the interpretation of the values are still not completely fixed.}

\subsection{A summary of the four options}
The readings may be summarized in the following correspondence table:
 \begin{center}
\begin{tabular}{c|cccc}
Options & $O1$ & $O2$ & $O3$ & $O4$ \\
\hline
$\mathbf{1}$ & $\mathbf{t}$ & $\mathbf{t}$ & $\mathbf{b}$ & $\mathbf{b}$\\
$\mathbf{0}$ & $\mathbf{f}$ & $\mathbf{f}$ & $\mathbf{n}$ & $\mathbf{n}$\\
$\mathbf{i}$ & $\mathbf{b}$ & $\mathbf{n}$ & $\mathbf{t}$ & $\mathbf{f}$\\
$\mathbf{j}$ & $\mathbf{n}$ & $\mathbf{b}$ & $\mathbf{f}$ & $\mathbf{t}$\\
\end{tabular}
\end{center}
There are at least four different ways of interpreting the truth-values for the system. Notice that this is not just a matter of re-interpreting them with different names, but the fact that \textit{different accounts of the truth values} may be exchanged, while, at the same time, the meaning of the connectives changes, without a difference being made at the level of the deductive rules for the system. 

As a result of this process of re-interpreting truth values and connectives, we obtain a strengthened version of the Carnap problem for the system under consideration. The matter is that a quite radical form of underdetermination arises that is internal to a fixed formal apparatus selected for the semantics; the same truth-values are able to exchange their roles and that change can go quite unnoticed from the point of view of the deductive behavior. This situation seems to be even more complex, or, at least, to add a layer of complexity to typical situations where categoricity is lacking, given that what is at stake here is not that additional surplus truth values are being added without making a difference for the consequence relation, rather a sort of underdetermination of meaning that can cause serious problems; radical misunderstanding can arise without being noticed.  

That is precisely where intelligibility seems to be threatened, and one needs some common background from which to access the many options. To discuss this, the idea that the common more or less classical background offered by the relational semantics may be of help (see also \cite{Omori2023-theoria} for additional discussion). That is, at least, the suggestion by Susan Haack, as we take it in the Haackian strategy, and as we shall discuss next.  

%\hitoshi{Maybe we can add an example of the mutual misunderstanding by taking a pair from the four options and focus on the Dunn semantic presentation? Jonas: I have added in the next section!}

%\hitoshi{The observation may be strengthened by adding a more general semantics in terms of bilattice-theoretic semantics so that we can say that the problem is much more substantial than in classical logic? This is because even if we pin down the number of the semantic values, the interpretation of the values are still not complete.}

\section{Haackian theme: another representation for classicists}\label{sec:Haack}
%Hitoshi, maybe the title is a bit out of place now, and we could put it as `understanding for classicists' or something like that
Although there is such a deep underdetermination of meaning between the four reading options, and one may be asking what kind of truth values are being dealt with here and what the connectives actually mean, there is also a sense in which, once an option is fixed, a classical logician can make sense of what is being advanced and of the kinds of disagreement that are at stake when it comes to deal with the other remaining readings. The first step for a better understanding of these problems concerns recognizing that not everything is lost once there are many conflicting options. As Susan Haack claimed, the first step is to notice that the addition of new truth values is \textit{not always} accompanied by the rejections of bivalence or two-valuedness:
\begin{quote}
    Not surprisingly, it has sometimes been supposed that the use of a many-valued logic would inevitably involve a claim to the effect that there are more than two truth-values [\dots] But in fact, I think it is clear that a many-valued logic needn't require the admission of one or more extra truth-values over and above `true' and `false', and indeed, that it needn't even require the rejection of bivalence. \cite[p.213]{haack1978} 
\end{quote}
The claim here is that one should avoid new or `sui generis' truth values that should be understood on their own terms, such as `paradoxical', or `meaningless', because the addition of such truth values is incompatible with the purported aims of logic (the investigation of which inferences are legitimate, in the sense of truth preservation) and they are, in the end, quite obscure if they are to  have a proper meaning, not couched in terms of usual truth values. With those exceptions out of the way, one may sometimes provide for readings of the new truth values that are compatible with two-valuedness, preserving the intelligibility of the system (see a discussion in \cite{OmoriISMVL22,Omori2023-theoria}). 

That strategy can certainly be applied in scenarios involving four truth values, as the one we have been discussing so far. That happens because truth values such as `neither true nor false' are actually just the lack of classical truth-values, and values such as `both true and false' just indicates that both classical values are attributed to a formula. So, in a sense, the truth values required here are not of a \textit{sui generis} kind. Adding `neither' and `both' is certainly not very classical, but their readings are clear enough to meet Haack's standards. Discussing the case of {\bf K3}, where the third truth value may be read as `neither', Haack explains:
\begin{quote}
    Assignment of the third truth value to a wff [wellformed
formula] indicates that it has no truth value,
not that it has a non-standard, third truth value. \cite[p.213]{haack1978}
\end{quote}

The situation for the four-valued system we are considering here, then, gets clearer when we use the Dunn semantics presented for the four options discussed above. There is then a common ground of truth values that allows one to make sense of the different readings of the system and of the meanings of the connectives. Notice, nothing in Haack's strategy requires that one and the same system cannot have different readings in such more intelligible terms, it is only required that at least one such reading exists.

%\jonas{Hitoshi: one option would be to move the relational semantics to this part of the text, to make the Haackian theme more emphatic and, perhaps, a bit surprising.}

One can make the case for the difference in understanding of the meanings clearer by selecting, for instance, options 1 and 2 (a similar problem arises for any pair of options, of course). For fixing on a more specific problem, let us be concerned with negation according to these two options. This gives rise to problems that may look quite similar to cases available in the literature concerning the similarity of gaps and gluts (see \cite{Beall2004}). Suppose that two adherents, one of option $1$, and another one of option $2$, agree that a given proposition $A$ is true (i.e. $1 \in V(A)$). The supporter of reading $1$ would claim that $\Not A$ is both $1$ and $0$, while the friend of reading $2$ would understand it as being neither $1$ nor $0$. The disagreement persists with $\Not \Not A$; the problem, however, is that they would happily agree that $A \lor \Not \Not A$. So, nothing changes, from a logical point of view if we focus only on the deductive behavior, although there is an abyssal difference in understanding the meaning of the truth values and the meaning of negation. We can, however, clearly make sense of the differences once the terminology of Dunn semantics is employed. A sort of common background is offered for discussion, although there is no \textit{purely logical} grounds for distinguishing those readings.

\section{Reflections}\label{sec:reflections}

\subsection{Meaning of the connectives}
%Theoria paper theme, as well as JC topical collection theme.
%\subsection{Left turn or right turn?}

%We can also view it as a variant of conflation-conjunction-disjunction fragment of Bilattice logic.

%\subsection{A more philosophy of math way of seeing the result?}
%Jonas suggested.

One of the most important themes in the philosophy of logic concerns the problem of whether changing a system of logic would require a corresponding change of the meaning of the connectives, leading to failure in legitimate rivalry between such systems. This is the famous \textit{meaning variance thesis}, commonly associated with Quine (see \cite[p.81]{Quine1986}, see also \cite{Paoli2003,Hjortland2022}), and it will be fruitful to present it here to elaborate on a contrast with the problem that we are highlighting in our paper. According to the meaning variance problem, for instance, if two systems disagree on the validity of the law of excluded middle, they may be understanding the meanings of the connectives involved in different terms: 
\begin{quote}
\dots the best explanation of this meaning change is that one or more of the logical constants occurring in the sentence have changed their meaning. This thought can be spelled out in a number of roughly equivalent ways, but all of them involve the idea that, for example, meaning what the classical logician means by ``not'' and ``or'' \textit{suffices} for acceptance of any instance of excluded middle whatsoever, at least potentially. So, if some particular instance of excluded middle isn't accepted, it must be because either ``not'' or ``or'' (or both) are being understood in a non-classical fashion. (\cite[p.423]{Warren2018})
\end{quote}
This is a difficult problem, and it is not even clear whether Quine himself would have endorsed the typical conclusion leading to scepticism about substantial disputes between different systems. It could actually be the case that there is such a radical meaning variation, while still it being the case that dispute concerning the appropriateness of choice of one of the systems may happen in the open, with disagreement about the validity specific laws and inferences; the case is that such dispute may be conduced according to a dispute on different reasons for accepting or rejecting a system involving the disputed laws and inferences, i.e. one may have a reason to prefer one system over the other. As Quine himself famously put: 
\begin{quote}
    [W]hoever denies the law of excluded middle changes the subject. This is not to say that he is wrong in so doing. In repudiating `$p$ or ${\sim}p$' he is indeed giving up classical negation, or perhaps alternation, or both; and he may have his reasons. \cite[p.83]{Quine1986}
\end{quote}
So, even though there may be disagreement on meaning, there is a dispute that can be conducted according to some kind of exchange of reasons pro and con each system. Quine mentions the simplification of quantum mechanics as one possible reason to revise classical logic (although he himself, of course, did not recommend taking that route). 

All of that is well known. However, that scenario provides for a nice platform from which to consider the problem we have been advancing here, which is way more radical. Given that we have at least four reading options for the same system, what results is that we actually have `change of subject', but without having the option to advance reasons, because we are not changing the logic. That is, there is a sense in which a change in the underlying meaning of the connectives does not carry over to a change in logic, so that one cannot carry the dispute with reasons for or against certain laws, because all the parties involved accept precisely the same logic (in fact, precisely even the same truth tables). So, the situation here is that we have change of meaning, without change of logic. One could not, for instance, argue that one reading is better because it simplifies quantum mechanics, while the other does not, given that from the point of view of logical consequence, any reading will do exactly the same as the others would.

The result is similar to Quinean scenarios of indeterminacy of translation; we may never be sure whether our understanding of the logical vocabulary of {\bf CNL$^2_4$} is the intended one by some other user. Someone whose knowledge of the meaning of the connectives were obtained exclusively from the familiarity with the derivation rules of the system would be able to learn completely incompatible lessons from the same teachings. Also, she would not be able to be sure that, whenever someone else uses the same logic, that someone is actually using the logic in the same meaning. This was illustrated before with the particular case of negation. 

That leads us directly to a worry that is also related to the meaning variance problem, which is the problem of determining whether the connectives we are labelling as conjunction, disjunction and negation are actually such logical connectives. Typically, the Quinean conclusion that a change of meaning engenders that we are no longer dealing with `the' logical connectives is quite well known. In a debate between a classical logician and a paraconsistent logician, remember, Quine famously wrote: 
\begin{quote}
My view of the dialogue is that neither party knows what he is talking about. They think that they are talking about negation, `$\sim$', `not'; but surely the notion ceased to be recognisable as negation when they took to regarding some conjunctions of the form `$p \ . {\sim}p$' as true, and stopped regarding such sentences as implying all others. Here, evidently, is the deviant logician’s predicament: when he tries to deny the doctrine he only changes the subject. (Quine \cite[p.81]{Quine1986})
\end{quote}
So, if we set aside the part of the comment involving failure of some inference (explosion), what Quine is saying is that negation sign is not a negation if it allows some contradictions to be true. That certainly applies to our case, where some readings of negation do allow for such scenarios, while others allow that negation changes truth values in ways incompatible with the expected truth to falsity behavior. That problem can also be extended to conjunction and disjunction, for sure. Probably, one could remark that a conjunction and a disjunction not satisfying the classical truth and falsity conditions are not the proper logical connectives. 

But once one adopts the thesis that meaning is defined in model theoretic terms, related to truth and falsity conditions, the possibility opens up for some kind of rescue of the two binary connectives in our four options. A classicist like Quine would require precisely the classical truth and falsity conditions for these connectives to be identified. However, once one leaves aside the assumption that the classical characterization is the correct one, still we can provide for a kind of minimal meaning conditions by restricting ourselves to the truth conditions for conjunction and disjunction. As one can easily check, the truth conditions for these connectives are held constant ---indeed, they are the classical conditions--- in all four options, leaving some flexibility for the falsity conditions (for more on the separation of such conditions while preserving the connective, see also the discussion in \cite{ArenhartOmori2024}). That is, if we can identify the connectives solely by their truth conditions, then, there is a sense in which these connectives are actually conjunction and disjunction. 

This is even clearer from the Haackian perspective that we are adopting here. Given the two-valued relational semantics, we can not only explain the differences in reading of the truth values in each case, but also make explicit the common truth conditions, and the diverging falsity conditions. There is a sense in which a classical part of the meaning of the connectives is preserved, with such truth conditions. It is not as easy to say the same about negation. Let us discuss it explicitly.

\subsection{Is $\sim$ a negation?}

The debate gets even more complicated when it comes to deal with negation: there is a question if $\Not$ is negation or not. Remember, if we follow the strategy made use of in \cite{HH2018contra,Omori2022varieties}, by applying the mechanical procedure described in \cite{OS-BD-FC}, then we obtain the following truth and falsity conditions, assuming that we rewrite the four values $\mathbf{t}, \mathbf{b}, \mathbf{n}$ and $\mathbf{f}$ as $\{ 1 \}, \{ 1, 0 \}, \emptyset$ and $\{ 0 \}$, respectively. 

Then, for options 1 and 4, we have the following conditions: 

\begin{tabular}{cc}
$1\in V(\Not A)$ iff $0\not\in V(A)$; & $0\in V(\Not A)$ iff $1\in V(A)$.
\end{tabular}

\noindent In view of these conditions, seen in the light of classicists' background assumption that holds for our Haackian strategy, the connective $\Not$ may be regarded as negation thanks to the falsity condition for negation. Remember, we are assuming that partial satisfaction of the classical truth conditions is enough for meaning attribution.

On the other hand, for options 2 and 3, we have the following conditions: 

\begin{tabular}{cc}
$1\in V(\Not A)$ iff $0\in V(A)$; & $0\in V(\Not A)$ iff $1\not\in V(A)$.
\end{tabular}

\noindent 
Now, that means the classical truth condition is satisfied. Again, that would ensure that $\Not$ can be regarded as negation by building on the classicists' understanding.\footnote{This is the understanding of negation taken, for example, in \cite{Avron2005,Omori2022varieties,ArenhartOmori2024}.}

Once the preferred option is settled, we are ready to interpret the results from Proposition~\ref{prop:valid-inferences}. Still, we need to take into account of the differences in the definition of the consequence relation. For example, options 2 and 3 do agree on the truth and falsity conditions for negation, but they disagree on the definition of the consequence relation since option 2 has non-falsity preservation in mind, whereas option 3 has truth preservation mind. Therefore, the first two items from the proposition will imply negation incompleteness of {\bf CNL$^2_4$} when option 2 is taken, whereas the same items imply negation inconsistency of {\bf CNL$^2_4$} for those preferring option 3.

If one requires more to regard a unary connective as negation, for example requiring both truth and falsity condition to be the same with classical logic, or even requiring a different truth and/or falsity condition, then the unary connective $\Not$ will not be regarded as negation.\footnote{This seems to be the direction pursued in \cite{BGZ2022}.} An alternative route to discuss the issue of negation is related to the possibility of defining different connectives that behave as classical negation inside the system.\footnote{For some discussions on classical negation in the context of {\bf FDE}, see \cite{DOBD+,SzmucOmori2022}.} 
%In that sense, we do not even need to consider merely partial meaning attributions. 
In this case, negation is a connective inside the system, but that seems to conclude that the actual connective $\Not$ is not a negation. In this case, it seems that the question of how to interpret the connective $\Not$ seems to remain.
%\hitoshi{TO JONAS: not sure what to add here, or even if it is necessary to add something more here. Thoughts are very welcome!}
%\jonas{Maybe we should keep the discussion concerning whether this is a negation to the subsection on the Reflections?}

%A suggestion (April 25): given the length of the paper so far, we could leave the discussion on hte meaning of the connectives out for the extended version (if we are  nvited to submit an extended version), and shift the negation discussion back to the section on the plurality of negations. With some massage, we make a very short note on how a classicist can seee negation here. Again, further discussionfor the extended version.

%We can contrast the non-classical vs. classical as well, not only within the non-classical theorists.

%\hitoshi{Building on various understandings, we can reflect less technically on the propositions and possibly connect it to the theme of negation inconsistency and negation incompleteness.}

\section{Concluding remarks}\label{sec:concluding-remarks}

In this paper, we have explored how themes from Carnap and Haack display in the system {\bf CNL$^2_4$} advanced by Oleg Grigoriev and Dmitry Zaitsev. The themes are a direct consequence of some re-workings we provided for the system. We have not only provided for an alternative proof system, but also discussed how  four options of readings for the semantics are available, connecting them with other systems available in the literature. The very idea that four readings are available for the original four truth values of {\bf CNL$^2_4$} gives rise to the Carnap problem: distinct approaches to truth and falsity are available and are put on the top of the same formal semantics. That gives rise to incompatible readings that still fit the system, in a sense that, broadly put, makes the system compatible with quite incompatible readings of the truth values it intends to deal with. Our proposed strategy to make sense of this diversity is to use Haack's claim that the use of a two-valued setting is appropriate to confer intelligibility to the system. That was achieved through a relational semantics, which made clearer, from a classical perspective, how the options differ. We have also discussed, on these lights, the status of negation in the system, connecting the truth and falsity conditions to the classicists' demands, but also identifying candidates for a classical negation inside the system, as provided by functional completeness.

\section*{Appendix}

\noindent Here are the details of the proof of Lemma~\ref{lem:canonical-valuation1}. For the case of negation, it goes as follows.
\begin{itemize}
\setlength{\parskip}{0cm}
\setlength{\itemsep}{0cm}
\item $v_\Sigma(\Not B){=}\mathbf{1}$ iff $v_\Sigma(B){=}\mathbf{j}$ (by the definition of $v_\Sigma$) iff $\Sigma\not\vdash B$ and $\Sigma\vdash \Not B$ (by IH) iff $\Sigma\vdash \Not B$ and $\Sigma\vdash \Not \Not B$ (by $(\Not\Not 2)$ for the left-to-right direction and $(\Not\Not 1)$ for the other direction).

\item $v_\Sigma(\Not B){=}\mathbf{i}$ iff $v_\Sigma(B){=}\mathbf{1}$ (by the definition of $v_\Sigma$) iff $\Sigma\vdash B$ and $\Sigma\vdash \Not B$ (by IH) iff $\Sigma\vdash \Not B$ and $\Sigma\not\vdash \Not \Not B$ (by $(\Not\Not 1)$ for the left-to-right direction and $(\Not\Not 2)$ for the other direction).

\item $v_\Sigma(\Not B){=}\mathbf{j}$ iff $v_\Sigma(B){=}\mathbf{0}$ (by the definition of $v_\Sigma$) iff $\Sigma\not\vdash B$ and $\Sigma\not\vdash \Not B$ (by IH) iff $\Sigma\not\vdash \Not B$ and $\Sigma\vdash \Not \Not B$ (by $(\Not\Not 2)$ for the left-to-right direction and $(\Not\Not 1)$ for the other direction).

\item $v_\Sigma(\Not B){=}\mathbf{0}$ iff $v_\Sigma(B){=}\mathbf{i}$ (by the definition of $v_\Sigma$) iff $\Sigma\vdash B$ and $\Sigma\not\vdash \Not B$ (by IH) iff $\Sigma\not\vdash \Not B$ and $\Sigma\not\vdash \Not \Not B$ (by $(\Not\Not 1)$ for the left-to-right direction and $(\Not\Not 2)$ for the other direction).
\end{itemize}
For the case of conjunction, it goes as follows.
\begin{itemize}
\setlength{\parskip}{0cm}
\setlength{\itemsep}{0cm}
\item $v_\Sigma(B\land C){=}\mathbf{1}$ iff $v_\Sigma(B){=}\mathbf{1}$ and $v_\Sigma(C){=}\mathbf{1}$ (by the definition of $v_\Sigma$) iff $\Sigma\vdash B$ and $\Sigma\vdash \Not B$ and $\Sigma\vdash C$ and $\Sigma\vdash \Not C$ (by IH) iff $\Sigma\vdash B\land C$ and $\Sigma\vdash \Not (B\land C)$. 

\item $v_\Sigma(B\land C){=}\mathbf{i}$ iff ($v_\Sigma(B){=}\mathbf{1}$ and $v_\Sigma(C){=}\mathbf{i}$) or ($v_\Sigma(B){=}\mathbf{i}$ and $v_\Sigma(C){=}\mathbf{i}$) or ($v_\Sigma(B){=}\mathbf{i}$ and $v_\Sigma(C){=}\mathbf{1}$) (by the definition of $v_\Sigma$) iff ($\Sigma\vdash B$ and $\Sigma\vdash \Not B$ and $\Sigma\vdash C$ and $\Sigma\not\vdash \Not C$) or ($\Sigma\vdash B$ and $\Sigma\not\vdash \Not B$ and $\Sigma\vdash C$ and $\Sigma\not\vdash \Not C$) or ($\Sigma\vdash B$ and $\Sigma\not\vdash \Not B$ and $\Sigma\vdash C$ and $\Sigma\vdash \Not C$) (by IH) iff $\Sigma\vdash B\land C$ and $\Sigma\not\vdash \Not (B\land C)$. 

\item $v_\Sigma(B\land C){=}\mathbf{j}$ iff ($v_\Sigma(B){=}\mathbf{1}$ and $v_\Sigma(C){=}\mathbf{j}$) or ($v_\Sigma(B){=}\mathbf{j}$ and $v_\Sigma(C){=}\mathbf{j}$) or ($v_\Sigma(B){=}\mathbf{j}$ and $v_\Sigma(C){=}\mathbf{1}$) (by the definition of $v_\Sigma$) iff ($\Sigma\vdash B$ and $\Sigma\vdash \Not B$ and $\Sigma\not\vdash C$ and $\Sigma\vdash \Not C$) or ($\Sigma\not\vdash B$ and $\Sigma\vdash \Not B$ and $\Sigma\not\vdash C$ and $\Sigma\vdash \Not C$) or ($\Sigma\not\vdash B$ and $\Sigma\vdash \Not B$ and $\Sigma\vdash C$ and $\Sigma\vdash \Not C$) (by IH) iff $\Sigma\not\vdash B\land C$ and $\Sigma\vdash \Not (B\land C)$. 

\item $v_\Sigma(B\land C){=}\mathbf{0}$ iff $v_\Sigma(B){=}\mathbf{0}$ or $v_\Sigma(C){=}\mathbf{0}$ or ($v_\Sigma(B){=}\mathbf{i}$ and $v_\Sigma(C){=}\mathbf{j}$) or ($v_\Sigma(B){=}\mathbf{j}$ and $v_\Sigma(C){=}\mathbf{i}$) (by the definition of $v_\Sigma$) iff ($\Sigma\not\vdash B$ and $\Sigma\not\vdash \Not B$) or ($\Sigma\not\vdash C$ and $\Sigma\not\vdash \Not C$) or ($\Sigma\vdash B$ and $\Sigma\not\vdash \Not B$ and $\Sigma\not\vdash C$ and $\Sigma\vdash \Not C$) or ($\Sigma\not\vdash B$ and $\Sigma\vdash \Not B$ and $\Sigma\vdash C$ and $\Sigma\not\vdash \Not C$) (by IH) iff $\Sigma\not\vdash B\land C$ and $\Sigma\not\vdash \Not (B\land C)$. 
\end{itemize}
The case for disjunction is similar to the case for conjunction.

%
%%\nocite{*}
%\bibliographystyle{eptcs}
%\bibliography{Oleg,Dunn}

\end{document}